\documentclass[preprint,12pt]{elsarticle}
\usepackage{amsmath}
\usepackage{amsthm}
\usepackage{amssymb}
\usepackage{amsfonts}
\newcommand{\Levy}{\L\'{e}vy \,}
\newcommand{\pit}{\pi_t}
\newcommand{\pis}{\pi_s}

\newcommand{\mut}{\mu_t}
\newcommand{\sigt}{\sigma_t}

\newcommand{\U}{\mathbf{U}}

\newcommand{\R}{\mathbf{R}}
\newcommand{\E}{\mathbf{E}}
\newcommand{\V}{\mathbf{Var}}

\newcommand{\calB}{\mathcal{B}}

\newcommand{\calN}{\mathcal{N}}

\newcommand{\calF}{\mathcal{F}}

\newcommand{\calA}{\mathcal{A}}
\newcommand{\half}{\frac{1}{2}}
\newtheorem{thm}{Theorem}

\newtheorem{prop}{Proposition}

\newtheorem{remark}{Remark}
\newtheorem{defn}{Definition}

\numberwithin{equation}{section}
\journal{Arxiv}

\begin{document}

\begin{frontmatter}

\title{A Mispricing Model of Stocks Under Asymmetric Information }

\author[buckley]{W. Buckley}
\author[brown]{G. Brown}
\author[marshall]{M. Marshall}

\address[buckley]{CBA, Florida International University, Miami, FL
  33199, USA.\\ Email: \texttt{wbuck001@fiu.edu}}
\address[brown]{Corresponding author: Statistical Laboratory, CMS, Cambridge University, CB3 0WB UK.\\
  Email: \texttt{gob20@statslab.cam.ac.uk}; } \address[marshall]{
  School of Management, University of Texas at Dallas, Richardson,
  Texas 75080, USA.\\ Email:
  \texttt{mario.marshall@student.utdallas.edu} }

\begin{abstract}
We extend the theory of asymmetric information in mispricing models for stocks following geometric Brownian motion to constant relative risk averse investors. Mispricing follows a continuous mean--reverting Ornstein--Uhlenbeck process. Optimal portfolios and maximum expected log--linear utilities from terminal wealth for informed and uninformed investors are derived. We obtain analogous but more general results which nests those of Guasoni (2006) as a special case of the relative risk aversion approaching one. 

\end{abstract}


\end{frontmatter}
\section{Introduction}
A simple mathematical model of two small investors on a financial market, one of which is better informed than  the other has attracted much attention in recent years. Their information is modelled by two different filtrations--the less informed investor has a sigma--field $ \calF^0$, corresponding to the natural evolution of the market, while the better informed investor has
the larger filtration $\calF^1$, which contains the information of the uninformed investor.
Understanding the link between asset mispricing and asymmetric information is a topic of ongoing interest in the finance literature.\\

Mispricing models for stocks under asymmetric information were first studied in Shiller (1981)
and Summers (1986) in a purely
deterministic setting. This was extended by Guasoni (2006) to the purely continuous random environment, 
where stock prices follow geometric Brownian motion (GBM) and utility is logarithmic. Fads/mispricing follow a continuous mean--reverting Ornstein--Uhlenbeck (O--U) process.  
There are two investors trading in the market--
the uninformed and informed investors.
Guasoni (2006) gives optimal portfolios and maximum expected logarithmic utilities, including
asymptotic utilities for both uninformed and informed investors.  He also gives  the excess asymptotic utility of the informed
investor. 

Buckley (2009) extends this theory to \Levy markets, where stock prices jump. Utility functions are assumed to be logarithmic. Jumps
are modelled by pure jump \Levy  processes, while the mispricing is represented by a purely continuous mean--reverting O--U
process driven by a standard Brownian motion, as in Guasoni (2006). 
The author obtains optimal portfolios and maximum expected logarithmic utilities for both the informed and uninformed 
investors, 
including asymptotic excess utility
  which is analogous to the result 
obtained in Guasoni (2006) in the purely continuous case. 
The random portfolios of the investors are linked to the symmetric,
purely deterministic optimal portfolios of L\'{e}vy diffusion markets having deterministic market coefficients.\\

In this paper, we generalize the theory of mispricing models of stocks under asymmetric information, where investors preference are from the power utility family.
We allow the stock price dynamic to move continuously as geometric Brownian motion, while the mispricing process remains as a continuous mean-reverting Ornstein--Uhlenbeck process. We obtain analogous but more general results which includes those of Guasoni (2006) as a special case of the risk aversion approaching one.  \\

The rest of the paper is organized as follows:
 Section 2 gives a brief literature review  and 
the model is introduced in Section 3. Filtrations are defined in Section 4, while price dynamics for both informed and uninformed investors are introduced in Section 5. Section 6 introduces CRRA (Constant Relative Risk Aversion) utility functions of the power class, while Section 7 presents portfolio and wealth processes. 
The main result is presented in Sections 8. We obtain optimal portfolios and log--linear maximum expected utilities for both the informed and uninformed 
investors, 
including asymptotic excess utility for the informed investor. 
 Section 9 concludes.
\section{Literature review}
\subsection{Continuous--Time Mispricing Models}
In this section, we give a brief literature review of asymmetric information in mispricing models in a purely continuous random market-- that is, in a 
market where stock prices and mispricing move continuously, without jumping. Discrete-time fads/mispricing models were first introduced in Shiller (1981) and Leroy and Porter (1981) as plausible alternatives to the efficient market/constant expected returns/discount rate assumption (cf Fama (1970)).


Studies by Flavin (1983), Kleidon (1986), and March and Merton (1986), criticize these findings based on the statistical validity of these volatility tests. However, other studies confirm the earlier findings of the variance bounds tests of Shiller (1981), and Leroy and Porter (1981). For example, West (1988) develops a stock market volatility
test that overcame these criticisms. West's inequality test prove that, if discount rates are constant, the variance of the change in the expected present discounted value of future dividends is larger when less information is used. He also finds that stock prices are too volatile to be the expected present discounted value of dividends when the discount rate is constant.\\

Campbell and Shiller (1987) find that when dividends and prices are non-stationary, they are co-integrated under the dividend discounted model, that is, there is a linear combination of the two that is stationary. Using cointegration and the VAR (vector autoregressive) framework, they also confirm the findings of Shiller (1981).

Given the failure of the discounted dividend model to explain stock price variations, some researchers introduced behavioural finance models as possible alternatives.  Summers (1986), and Cutler et al.(1990) introduce irrational/noise traders, and the slow response to changes in fundamentals. 
 DeLong et al. (1990),  suggest that noise trading in the market can increase price volatility, which impacts  the risk of investing in the stock market and the risk premium.  Campbell and Kyle (1993) also suggest that the existence of noise trading in the market can help explain the high volatility of stock prices. Daniel et al. (1998) develop a theory of mispricing based on investor overconfidence resulting from biased self-attribution of investment outcomes. \\

  Barberis et al. (1998) provide an explanation for over and under-reactions based on a learning model in which actual earnings follow a random walk but individuals believe that earnings follow a steady growth trend, or are mean reverting.  Odean (1998) provides a model where overconfident traders can cause markets to under-react to the information of rational traders,leading to positive serially correlated returns. \\

Wang (1993) gives a model of intertemporal/continuous--time asset pricing under asymmetric information. In his paper, investors have different information concerning the future growth rate of dividends, which satisfies a mean--reverting Ornstein--Uhlenbeck process. Informed investors know the future dividend growth rate, while uninformed investors do not. 
All  investors observe current dividend payments and stock prices. The growth rate of dividends determines the rate of appreciation of stock prices, and stock price changes provide signals about the future growth of dividends. Uninformed investors rationally extract information about the economy from prices, as well as dividends.
 Wang (1993) shows that asymmetry among investors can increase price volatility and negative autocorrelation in returns; that is, there is mean--reverting behaviour of stock prices. Thus, imperfect information of some investors can cause stock prices to be more volatile than in the symmetric case, when all investors are perfectly informed.\\

 Brunnermeier (2001) presents an extensive review of asset pricing models under asymmetric information mainly in the discrete setting. He shows how information affects trading activity, and that expected returns depend on the information set or filtration of the investor. These models show that past prices still carry valuable information, which can be exploited using technical/chart analysis, which uses part or all of past prices to predict future prices.\\

 Guasoni (2006) extends the model of Summers (1986) to the purely continuous random setting. He develops models of stock price evolution for two disjoint classes of investors; the informed and uninformed investors. The informed investor, indexed by $i=1$, observes both the fundamental and market values of the stock, while the so--called uninformed investor, indexed by $i=0$, observes market prices only. Both investors have filtrations or information banks $\calF^i, \, i \in \{0,1\}$ with 
$\calF^0 \subset \calF^1\subset \calF, $ where $\calF$ is a fixed sigma--algebra. 
The problem of the maximization of expected logarithmic utility from terminal wealth is solved for each investor, and an explicit 
formula for the asymptotic excess utility of the informed investor takes the form $ \frac{\lambda}{2}p\,(1-p)\,T$ where $T$ is the long run investment horizon, $\lambda$ is the reversion speed, and $q^2=1-p^2$ is the proportion of mispricing in the market. \\

Buckley (2009) extends this theory to include stock prices that jump. Utility functions are still assumed to be logarithmic. Jumps
are modelled by pure jump \Levy  processes, while mispricing is represented by a purely continuous mean--reverting O--U
process driven by a standard Brownian motion. 
The author obtains optimal portfolios and maximum expected logarithmic utilities for both the informed and uninformed 
investors, 
including asymptotic excess utility of the form $\frac{\tilde{\lambda}}{2}\,p\,(1-p)\,T$, which is analogous to the result 
obtained by Guasoni (2006) in the purely continuous case.  The random portfolios of the investors are linked to the symmetric,
purely deterministic optimal portfolios of  L\'{e}vy diffusion markets having deterministic market coefficients. 

\section{The Model}
The model consists of two assets;\,a riskless asset \textbf{B} called bond, bank account or money market, and a risky asset $S$ called stock. The bond earns a continuously compounded risk--free interest rate $r_t$, while the stock has total percentage appreciation rate or \textbf{expected returns} $\mut$, at time $t\in[0,T]$.  The stock is subject to volatility $ \sigma_t>0$. The market parameters are $\mut,\, r_t,\, \sigma_t,\, $, and are \textbf{deterministic} functions. \\

\noindent\textbf{Standing Assumptions :}\\
(1)  $T>0$, is the investment horizon; all transactions take place in  $[0, T].$\\
(2)  The market parameters $ r,\, \mu,\, \sigma^2$ are Lebesgue integrable deterministic functions.\\
(3)  The stock's \textbf{Sharpe ratio } or \textbf{market price of risk} $\theta$, is square integrable.\\
(4)  The risky asset $S$ lives on a probability space $(\mathbf{\Omega},\, \calF,\, \mathbf{P})$ on which is defined two independent 
standard Brownian motions $W=(W_t)_{t\geq 0}$ and $B=(B_t)_{t\geq0}$.  $\calF$ is an sigma--algebra of subsets of $\mathbf{\Omega}$, and 
$\mathbf{P}$ is the ``real--world" probability measure on  $\calF$. \\
(5)  Fads or mispricing are modelled by the a mean--reverting \textbf{Ornstein--Uhlenbeck} (O--U) process $U=(U_t)_{t\geq 0}$ with 
mean--reversion rate or speed $\lambda. $ \\
(6)  Utility are from the power class; i.e., linear functions of $ U(x)= x^\gamma$, where $\gamma < 1$.\\ 
(7)  Informed and uninformed investors are represented by the indices ``1", and ``0", respectively. \\
\noindent \textbf{The Price Dynamic}\\
The bond $\textbf{B}$ has price 
$ \mathbf{B}_t= \exp{\left(\int_0^t r_s ds\right)},$
while the stock has log return dynamic 
\begin{eqnarray}
d(\log{S_t}) &=& (\mut - \half \sigma^2_t)dt+ \sigma_t dY_t,\:\:\: t\in[0,T],\label{2.1}\\
Y_t &=& p\,W_t + q\,U_t,\:\:\: p^2 + q^2=1, \;p\geq 0, \, q\geq 0,\label{2.2}\\
dU_t &=& -\lambda U_t dt + dB_t, \:\:\: U_0=0,\,\,\,\lambda > 0.\label{2.3}
\end{eqnarray}
Applying It\^{o}'s transformation formula to (\ref{2.1}) gives percentage return for the stock:
\begin{equation}
\frac{dS_t}{S_t}= \mut dt+ \sigma_t dY_t ,\:\:\: t\in[0,T].\label{2.1a}
\end{equation}
Observe that $\mut$ is the expected percentage return on the stock,  while $\sigma_t dY_t$ is the \textbf{excess} percentage returns. The fads or mispricing process $U$ is a mean--reverting 
Ornstein--Uhlenbeck process with speed $\lambda$, which is the unique solution of the \textbf{Langevin stochastic differential 
equation:} (\ref{2.3})
with explicit solution
\begin{equation}\label{2.9}
U_t =  \int_0^t e^{-\lambda (t-s)} dB_s,\:\:\: t\in [0,T],
\end{equation}
with mean $\,\,\E[U_t]=0$, and variance
$ \mathbf{Var}[U_t]=\frac{1-e^{-2\lambda t}}{2\lambda}. $\\
If the speed $\lambda$ is close to zero, mean reversion is slow and there is a high likelihood of mispricing, while if $\lambda>>0$, the 
mispricing reverts rapidly, thereby reducing any advantages of mispricing. \textbf{$100 q^2$ \% is the percentage of mispricing in the market.}
Equation (\ref{2.1}) or equivalently, (\ref{2.1a}), has unique solution
\begin{equation}\label{2.5}
S_t= S_0\,\exp{\left( \int_0^t(\mu_s - \half \sigma^2_s)ds+ \int_0^t\sigma_s dY_s\right)},\:\:\:t\in [0,T].
\end{equation}
By imposing (\ref{2.3}) on (\ref{2.2}), we see that $Y$ is a combination of a martingale $W$, which represents permanent price shocks, 
and $U$ the mean--reverting O--U process, which represents temporary shocks. If $\lambda=0$ or $q=0$, (and $\mut$ and $\sigt$ are constants) 
we revert to the usual geometric Brownian motion (GBM) of Merton (1971).
\noindent 
\section{Filtrations or Information Flows of Investors}
\begin{defn}\label{def2.1}
Let $X=(X_t)_{t \geq 0}$ be a process defined on $(\mathbf{\Omega},\, \calF,\, \mathbf{P})$. Its natural filtration $ \calF^{X} 
=(\calF_t^{X})_{t \geq 0}$ is the sub--$\sigma$ algebra of $\calF$ generated by $X$, and is given by 
\begin{equation*}\label{2.10}
\calF_t^{X}\stackrel{\triangle}{=}\sigma(X_s: s\leq t)=\{X^{-1}(A)\subset \mathbf{\Omega}: A \in \calB(\R)\}.
\end{equation*}
\end{defn}
\noindent $\calF_t^{X}$ is the information generated by $X$ up to time $t$. 
\subsection{Augmentation}
We can make $\calF_t^{X}$ right--continuous and complete by augmenting it with $\calN$, the $\mathbf{P}$--null sets of $\calF$, given by
$\calN = \{ A \subset \mathbf{\Omega}: \exists\, B\in \calF,\:\: A\subset B,\:\: \mathbf{P}(B)=0\}.$
\noindent The augmented filtration of $X$ is $\sigma\,(\calF^{X}\vee \calN)$. $\calF^{X}$ is \textbf{complete} if it contains the $\mathbf{P}$--null sets of $\calF$, e.g.,  if $\calN \subset \calF_0^{X}$.
A filtration $(\calF_t)$ is \textbf{right continuous} if 
$\calF_{t+}=\calF_t, \:\:\:\:\:\mbox{where\,\, } \calF_{t+}=\bigcap_{s >t} \calF_s.$

\noindent In the sequel, we assume that all filtrations $(\calF_t)$ are right--continuous and complete. In this case, we say the filtration satisfies the \textbf{usual hypothesis} (Applebaum, 2004). Thus, $\calF^{X}=(\calF_t^{X})_{t \geq 0}$ will denote the complete right--continuous  filtration generated by $X$ on $(\mathbf{\Omega},\, \calF,\, \mathbf{P})$.
The informed investor observes the pair $(S,\,U)$, while the uninformed investor observes only the stock price $S$. 
\begin{defn}[\textbf{Filtrations of Investors}]\label{def2.13}
Let $\calF^1=(\calF^1_t)_{t \geq 0}$ and $\calF^0=(\calF^0_t)_{t \geq 0}$ be the filtrations generated by $(S,\,U)$ and $S$, 
respectively. That is, for each $t\in [0, T]$
\begin{equation*}
\calF^1_t \stackrel{\triangle}{=}\sigma\,(S_s,\,U_s : s\leq t)=\calF^{S,\,U}_t\,\,\,\mbox{and}\,\,\,\,
\calF^0_t \stackrel{\triangle}{=}\sigma\,(S_s : s \leq t)=\calF^{S}_t\label{2.13s}.
\end{equation*}
\end{defn}
\noindent $\calF^1$ and $\calF^0$ are the respective information flows of the informed and uninformed investors. 
Equivalently, since $W$ and $B$ generate $S$ and $U$ for the informed investor, while $Y$ generates $S$ for the uninformed investor, 
then
\begin{equation*}
\calF^1_t \stackrel{\triangle}{=}\sigma\,(W_s,\,B_s : s \leq t)=\calF^{S,\,U}_t\,\,\,\mbox{and}\,\,\,\,
\calF^0_t \stackrel{\triangle}{=}\sigma\,(Y_s:s \leq t)=\calF^{Y}_t\label{2.14s}.
\end{equation*}
Clearly
$\calF^0 \subset\calF^1 \subset\calF \:\:\:\Longleftrightarrow \calF^0_t\subset \calF^1_t$ for all $\:\:\; t\in [0,T].$
The market participants can be classified in accordance to their respective information flows. Those with access to $\calF^1$, are 
called informed investors--they observed both the fundamental and market prices of the risky asset. 
Those with access to $\calF^0$ only, 
are called uninformed investors--they observe market prices only. These uninformed investors know that there is mispricing in the market but 
cannot observe them directly.

\section{The Stock Price Dynamic for the Investors}
It follows from (\ref{2.1a}), that for  both investors the general percentage return for the stock has dynamic
$\frac{dS_t}{S_t}= \mut\, dt+ \sigma_t\, dY_t.$
\noindent We rewrite this dynamic for each investor.
\subsection{Price Dynamic for the Uninformed Investor}

Using the Hitsuda (1968) representation of Gaussian processes (see Cheridito,2003), Guasoni (2006) gives an $\calF^0$--Brownian motion $B^0$ and a random process $\upsilon^0$, such that $dY_t=dB^0_t + \upsilon_t^0dt  $. We now invoke a useful part of Theorem 2.1, Guasoni (2006). \\
Let $(\mathbf{\Omega},\, \calF,\, \mathbf{P})$ be a probability space on which independent Brownian motions $W$ and $B$ are defined. Let $\calF^0=(\calF^0_t)_{t \geq 0} \equiv (\calF^{Y}_t)_{t \geq 0}$ be the filtration generated by $Y$ satisfying the usual hypothesis. Let
\begin{equation}\label{2.24}
\gamma(s)=  \frac{1-p^2}{1+p \tanh( p \lambda s)}- 1.
\end{equation}
Then, we can construct an $\calF^0$-- Brownian motion $B^0=(B^0_t)$ on $(\mathbf{\Omega},\, \calF,\, \mathbf{P})$ such that in terms of \:\:$Y_s: s\leq t$, 
\begin{equation}\label{2.30}
Y_t = B^0_t + \int_0^t \upsilon^0_s ds,
\end{equation}
where $\upsilon_t^0$ is given by 
\begin{equation}\label{2.23}
\upsilon_t^0 \stackrel{\triangle}{=} -\lambda\, \int_0^t e^{-\lambda\, (t-s)} (1+\gamma(s))\,dB^0_s.
\end{equation}


\begin{remark}
$\gamma(u)$ is the solution of the equation:
$\gamma'(s)=\lambda (\gamma^2(s) - p^2),\:\:\: \gamma(0)= - p^2.$
\end{remark}
\subsection{Price Dynamic for the Investors}
The stock price dynamic of each investor, relative to his or her filtration, now follows.
\begin{prop}\label{pro1}
Let $i\in\{0,\,1 \}$. Under $\calF^i$, the price dynamic of the $i$--th investor is 
\begin{equation}\label{2.21.0}
\frac{dS_t}{S_t}=\mut^i dt + \sigma_t dB^i_t,
\end{equation}
with price
\begin{equation}\label{2.26}
S_t= S_0\,\exp{\left( \int_0^t(\mu^i_s - \half \sigma^2_s)ds+ \int_0^t\sigma_s dB^i_s\right)},\:\:\:t\in [0,T],
\end{equation}
where the stock has random drift
\begin{equation}\label{2.22}
\mut^i\stackrel{\triangle}{=}\mut+ \upsilon^i_t \sigma_t,
\end{equation}
 with  $B^1_t = p\,W_t + q\,B_t$, \,\,$\upsilon^1_t= -q\lambda\,U_t$, and $\upsilon_t^0$,\,\, $ B_t^0$ are defined in Equations (\ref{2.30}--\ref{2.23}).
\end{prop}
\begin{proof}
See Appendix A for proof.
\end{proof}
\begin{remark}
Observe from (\ref{2.21.0}) that for the $i$--th investor, that is, relative to the filtration $\calF^i$, the drift of the stock price is $\mut^i $, is \textbf{random}, while they both share a common deterministic volatility $\sigma_t $, which is the volatility of the original driving process given by (\ref{2.1a}), which has a deterministic drift $\mut$. Moreover, $\E \mut^i=\mut$
\end{remark}
We state without proof, a useful result for $\upsilon^i_t$ that will be required in the sequel.
\begin{prop}\label{pro2}
Let $i\in\{0,1\}$, $p \in [0,\,1]$, $p^2+q^2=1$ and  $t\in [0,\,T]$.\\
Let $\upsilon^0_t = -\lambda \int_0^t e^{-\lambda (t-s)} (1+\gamma(s))dB^0_s$ and  $\upsilon^1_t = -\lambda \,q\,U_t$. Then \\
$(0)$\hspace{.1in}  $\E[\upsilon^i_t]=0$. \\
$(1)$\hspace{.1in} $ \E[\upsilon^0_t]^2=\lambda^2 \int_0^t e^{-2\lambda (t-s)} (1+\gamma(s))^2 ds\le \half \lambda \,q^2.$\\  
$(2)$\hspace{.1in} $ \E[\upsilon^1_t]^2= \frac{\lambda}{2}\,(1-p^2)\,(1- e^{-2\lambda t}) \le \half \lambda \,q^2.  $\\ 
$(3)$\hspace{.1in} $ \E[\upsilon^i_t]^2= \frac{\lambda}{2}(1- p)(1+(-1)^{i+1}p) ,\:\:\:\mbox{ as } t\rightarrow \infty .$\\
$(4)$\hspace{.1in} $\int_0^T \E[\upsilon^i_t]^2 dt \simeq \frac{\lambda}{2}\,(1-p)(1+(-1)^{i+1}p)\,T, \:\:\: \mbox{as }\:\: T\longrightarrow \infty.$\\
$(5)$\hspace{.1in} As $T\longrightarrow \infty$, the asymptotic excess cumulative variance of the $v^i$s is
$$ \int_0^T \E[\upsilon^1_t]^2 dt - \int_0^T \E[\upsilon^0_t]^2 dt\simeq \lambda\,p\,(1-p)\,T\leq \frac{\lambda}{4}\,T.$$
\end{prop}
\begin{remark}
$(0)$ Observe that for the $i$--th investor, the random process $\upsilon^i_t$ has a mean of 0, and is simply the number of standard deviations (volatility units) of its random drift ($\mut^i $) given in its price dynamic from its deterministic mean ($\mut$). That is, $\E \mut^i=\mut$\\
$(1)$
Moreover, $\upsilon^i_t$ is an integrated Brownian motion, explicitly given by
\begin{eqnarray}
\upsilon^0_t &=& -\lambda \int_0^t e^{-\lambda (t-s)} (1+\gamma(s))dB^0_s, \nonumber\\
\upsilon^1_t & =& -\lambda \,q\,U_t= -\lambda \,q\,\int_0^t e^{-\lambda (t-s)} dB_s,\:\:\: t\in [0,T].\label{2.2215} 
\end{eqnarray}
$(2)$ The variance of the process $\upsilon^i_t$ is $ \E[\upsilon^i_t]^2$, a monotonically increasing function of $t$, the mean-reversion speed $\lambda$, and the mispricing level $q^2$. The variance is bounded by half of the reversion speed and goes to zero if any of these quantities approach zero. \\
$(3)$ The cumulative variances are also bounded above by the mean-reversion speed, the mispricing level $q^2$, and the investment horizon $T$.

\end{remark}

\section{Power Utility Functions}
We assume that each investor has a utility function $\U:(0,\,\infty)\rightarrow \R$ for wealth that  satisfies the Inada condition, i.e., 
 it is strictly increasing, strictly concave, continuously differentiable, with 
$$ \U'(0)=\lim_{w\downarrow 0}\U'(w)=+ \infty, \hspace{.5in}  \U'(\infty)=\lim_{w\rightarrow \infty}\U'(w)=0.$$
\noindent $\U_0(w)=\log{w}$, the logarithmic utility and $\U_\gamma(w)=\frac{w^\gamma}{\gamma},\:\: \gamma<1$, the power utility, satisfy 
this condition. In the sequel, all utility functions are assumed to be power, with constant relative risk aversion (RRA), $1-\gamma$. In particular, it is easy to show that 
$\U_\gamma(w)=\frac{w^\gamma -1}{\gamma}\longrightarrow \,U_0(w)=\log{w},\;\;\mbox{when}\, \gamma\longrightarrow0.$ \label{2.36aa}
\section{Portfolio and Wealth Processes of Investors}
\begin{defn}[\textbf{Portfolio Process}]
A portfolio process $\pi: [0,\,T]\times \Omega\rightarrow \R$, is an $\calF= (\calF_t)_{t\geq 0}$--adapted process satisfying
$\int_0^T (\pit \sigt)^2 dt < \infty,\,\,\, \mbox{ almost surely.}$
\end{defn}
\noindent Although $\pi$ is a function of $(t, \omega)$,  in the sequel we keep $ \Omega$ in the background, and assume that $\pi$ is 
primarily a function of time $t$, where $\pit$ is the proportion of an investor's wealth invested in the stock at time $t$.
 The remainder
$1-\pit$, is invested in the bond or money market.  $\pi$ is not restricted to [0,\,1] for the purely continuous model--we allow short--selling ($\pi<0$) and borrowing ($\pi>1$) at the risk--free rate.
\begin{defn}[\textbf{Self--financing}]
A portfolio process $\pi$ is called self--financing if 
\begin{equation}\label{2.37a}
dV_t= (1-\pit) r_t V_t\,dt + \pit V_t \frac{dS_t}{S_t},
\end{equation}
where $V_t$ is the wealth or value of the holding of stock and bond at time $t \in [0,T]$.
\end{defn}
\noindent Thus, for self--financing portfolios, the change in the wealth is due only to the change in prices, \textbf{provided} that no money is 
brought in or taken out by the investor.  \\

\noindent \textbf{The Wealth Process:}
For a given non--random initial wealth $x>0$, let $ V^{x,\,\pi}\equiv V^{\pi} \equiv V =(V_t)_{t\geq 0}$ denote the wealth process 
corresponding to a self--financing portfolio $\pi$ with $V_0=x$, and satisfying the stochastic differential equation (\ref{2.37a}).

\noindent We now present wealth dynamics for each investor.
\begin{thm}\label{th2.3}
Let $i\in\{0,1\}$ and let $r_t$ be the risk--free interest rate. Let $\pi^i$ and $V^i$ be the respective portfolio and wealth processes for the $i$--th investor as a result of investing in the stock, with Sharpe ratio
\begin{equation}\label{2.50}
\theta^i_t=\frac{\mut^i - r_t}{\sigma_t},\:\:\: \mut^i= \mut + \upsilon^i_t\, \sigma_t,\:\: t\in[0,\,T],
\end{equation}
and percentage return dynamic driven by an $\calF^i$--adapted Brownian motion $B^i$, given by (\ref{2.21.0}).
Then the wealth process $V^i=V^{i,\,\pi}$ corresponding to $\pi$, and initial wealth $\,x>0$, has percentage return dynamic
\begin{equation}\label{2.52}
\frac {dV^i_t}{V^i_t}=(r_t +\pit^i \sigma_t\theta^i_t)\,dt + \pit^i \,\sigma_t dB^i_t ,
\end{equation}
with unique discounted wealth process
\begin{equation}\label{2.53}
\widetilde{V}^i_t=x\, \exp \left( \int_0^t (\pis^i \sigma_s \theta^i_s -\half (\pis^i)^2 \sigma^2_s)ds + \int_0^t \pis^i \sigma_s dB^i_s 
\right)
\end{equation}
and under power utility $\U(w)=w^\gamma= e^{\gamma\,\log w},\:\:w>0,\, \gamma<1$,
\begin{equation}\label{2.55aa}
\U({\widetilde{V}^i_t})= e^{\gamma\,\log{\widetilde{V}^i_t}} = e^{\gamma\,\log{x} + \gamma\,H^{i}(t)} 
\end{equation}
where 
\begin{equation}\label{2.55ab}
H^{i}(t)=\half \int_0^t (2 \theta_s^i \pis^i \sigma_s - ( \pis^i \sigma_s)^2 )ds + \int_0^t \pis^i 
\sigma_s dB^i_s.
\end{equation}
\end{thm}
\begin{proof} 
See appendix A  for proof. \qedhere
\end{proof}
\section{ Utility Maximization from Terminal Wealth}
Each investor is assumed to be rational; that is, the investor is a utility maximizer. Thus, both informed and 
uninformed 
investors  maximize their respective expected utility from terminal wealth $V_T$, where $T$ is the investment horizon. 
The terminal wealth $V_T$ is represented by its discounted value $\widetilde{V}_T$. We then maximize $\E_{\pi} \U(\widetilde{V}_T)$ where $\pi$ 
is selected from an admissible set $\calA(x)$.
\begin{defn}[\textbf{Admissible Portfolio}]
A self--financing portfolio $\pi$ is admissible if $V^{\pi}_t$ is lower bounded for all $t\in [0,\,T].$
That is, there exists $ K> -\infty$ such that almost surely, $V^{\pi}_t > K $ for all $t\in [0,T].$-(cf Oksendal 2005, page 265)
\end{defn}
\noindent Since $x>0$, we assume that $V_t^{\pi} >0$ and therefore $ \widetilde{V}_t^{\pi}>0$ for all $t\in[0,\,T]$. Thus equivalently, $\pi$ is 
admissible if  $ \widetilde{V}_t^{\pi}>0$ for all $t\in[0,T]$.
 Karatzas and Shreve (1991), define an admissible portfolio in terms of the utility function $\U(x)$ by the prescription:
$\E[\U(\widetilde{V}_t^{\pi})]^{-} <  \infty,$
where $a^{-}= \max \{0,-a\}$. Either definition will suffice!
\begin{defn} [ \textbf{Admissible set}]
Let $x>0$ be the initial wealth of the investor. The admissible set $ \calA(x)$ of this investor is defined by 
\begin{equation*}\label{2.57}
\calA(x)=\left\{\pi: \pi-admissible, S-integrable, \calF-predictable \right\}.
\end{equation*}
\end{defn}
\noindent  $\pi$ is $\calF$--predictable if it is measurable relative to the predictable sigma--algebra on $[0,\,T]\times \Omega$
(see Protter(2004) for details). 
\subsection{Utility Maximization Problem and Optimal Portfolios}

For a given utility function $U(\cdot)$ and initial wealth $x>0$, we maximize the expected utility from (discounted) terminal wealth 
$\E[\U(\widetilde{V}_t^{\pi})]$, over the investors admissible set $\calA(x)$. The value function for this problem is 
$u(x)\stackrel{\triangle}{=} \sup_{\pi \in \calA(x)} \E[\U(\widetilde{V}_t^{\pi})]$,
where it is assumed that $u(x)<\infty$ for all $x>0$. That is, there is an optimal portfolio $\pi^* \in \calA(x)$ such that $u(x)=\E[\U(\widetilde{V}_t^{\pi^*})]$. In other words, suppressing superscripts for investors,
$\pi_t^*= \arg \sup\E[\U(\widetilde{V}_t^{\pi^*})].$
Let $i \in \{0,1\}$. For the $i$--th investor, define an admissible set:
$\calA^i(x)=\left\{\pi: \widetilde{V}_t^{\pi}> 0, \:a.s., S-integrable, \calF^i-predictable \right\}$
and a utility maximization problem: $ \sup_{\pi} \left\{\E [\U( \widetilde{V}_t^{\pi})]:  \pi \in \calA^i(x)\right\},$
with respective value functions:
\begin{equation}\label{2.63}
u^i(x)=\sup_{\pi}\left\{\E [\U( \widetilde{V}_t^{\pi})]:  \pi \in \calA^i(x)\right\}
=\E \U( \widetilde{V}_t^{\pi^{*,\,i}}).
\end{equation}

The logarithmic utility function is used by Guasoni (2006) so that an explicit solution of (\ref{2.63}) is obtained (cf  Amendinger et al.(1998), Imkeller and Ankirchner (2006), Karatzas and Pikovsky (1996)).  
We now give a slightly modified version of Guasoni's solution to (\ref{2.63}) using the notation developed in preceding sections.

\begin{thm}[ Guasoni (2006), Theorem 3.1]\label{th2.4}
Let $i\in\{0,1\}$,  $u^i(x)$ be the value function, and $\pi^{*,\,i}$ the optimal portfolio for the $i$-th investor that solves (\ref{2.63}), where utility is assumed to be logarithmic and the risk--free interest rate is $r=0$.\\
$(1)$\hspace{.1in} The optimal portfolio for the $i$--th investor is:
\begin{equation}\label{2.64}
\pit^{*,\,i}=\frac{\theta^i_t}{\sigma_t}=\frac{\mut^i}{\sigma^2_t}=\frac{\mut+\upsilon^i_t \sigma_t}{\sigma^2_t}, \;\;\;t\in[0,T],
\end{equation} 
where $\theta^i$ is the Sharpe ratio of the stock for the $i$--th investor.\\
$(2)$ \hspace{.1in} The maximum expected utility from terminal wealth for the $i$--th investor is 
\begin{eqnarray}
u^i(x) &= &\log{x}+ \half \E \int_0^T (\theta^i_t)^2 dt = \log{x}+\half \E \int_0^T \left(\frac{\mut+\upsilon^i_t \sigma_t}{\sigma_t}\right)^2 dt\label{2.65}
\end{eqnarray}
$(3)$\hspace{.1in} As $T\longrightarrow \infty$, the asymptotic maximum expected utility is 
\begin{equation}\label{2.67}
u^i_{\infty}(x)\simeq \log{x} + \half \int_0^T \frac{\mut^2}{\sigma^2_t}\, dt + \frac{\lambda}{4}(1-p)
(1+(-1)^{i+1}p)\,T.
\end{equation}
$(4)$\hspace{.1in} The excess asymptotic maximum expected utility of the informed investor is 
\begin{equation}\label{2.68}
u^1_{\infty}(x)- u^0_{\infty}(x) \simeq \frac{\lambda}{2}\,p(1-p)\,T.
\end{equation}
\end{thm}
\begin{remark}
The optimal portfolio for each investor, relative to its filtration, is its Sharpe ratio divided by the common volatility of the driving Brownian motion. It is random, being depended on the random processes $\upsilon^i_t$. The expected optimal portfolio is the deterministic ratio of the mean return ( $ \mut$) and the variance ($\sigma_t^2$) of the stock. \\
From (2), we observed that optimal utility is directly related to the natural logarithm of the initial investment $x$.
In (3), the asymptotic expected utility for each investor depends on $p=\sqrt(1-q^2)$ and hence on $q^2$, the proportion of mispricing in the stock price. It therefore follows 
 in (4), that the excess asymptotic utility which is a function of $p$, also depends on $q^2$, the proportion of mispricing in the stock price.
\end{remark}
We now give an analogous result to Guasoni (2006), starting with the optimal portfolio. 
\begin{prop}\label{pro2.22a}
Let $i\in\{0,1\}$, and $\pi^{*,\,i}$ the optimal portfolio for the $i$-th investor that solves (\ref{2.63}), where utility is the basis power function $ w^\gamma$ with RRA $1-\gamma$, and the risk--free interest rate is $r=0$. \\
$(1)$\hspace{.1in} The optimal portfolio for the $i$--th investor is:
\begin{equation}\label{2.64}
\pit^{*,\,i}=\frac{{\theta_\gamma}^i(t)}{\sigma_t}=\frac{\mut^i}{(1-\gamma)\sigma^2_t}=\frac{\mut+\upsilon^i_t \sigma_t}{(1-\gamma)\sigma^2_t}, \;\;\;t\in[0,T],
\end{equation} 
where $\theta_\gamma^i= \frac{\theta^i}{1-\gamma}$ is the risk-adjusted Sharpe ratio of the stock owned by the $i$--th investor.
\end{prop}
\begin{proof} \textbf{\{Proof of Proposition \ref{pro2.22a}\}}\\
From Theorem \ref{th2.3}, 
$ H^{i}(t)= \half \int_0^t( 2 \theta_s^i \pis^i \sigma_s - ( \pis^i \sigma_s)^2 )ds+ \int_0^t \pis^i 
\sigma_s dB^i_s. $  
 Therefore, \\$ \E (H^{i}(t)) = \half \E \int_0^t( 2 \theta_s^i \pis^i \sigma_s - ( \pis^i \sigma_s)^2 )ds $, 
 and by Ito's isometry, its  variance is
  $ \V (H^{i}(t)) = \E (\int_0^t \pis^i \sigma_s dB^i_s)^2 = \int_0^t (\pis^i \sigma_s)^2 ds. $ 
 Whence, $\E \U({\widetilde{V}^i_t})= \E e^{\gamma\,\log{x} + \gamma\,H^{i}(t)}\\= e^{\gamma\,\log{x} + \E(\gamma\,H^{i}(t)) + \half \V (\gamma\,H^{i}(t))} = e^{\gamma\,\log{x} + \gamma \, \E (H^{i}(t)) + \half \gamma^2\,\V (H^{i}(t))} $. Since the exp function is increasing,
  $  u^i(x)= \sup e^{\gamma\,\log{x} + \gamma \, \E (H^{i}(t)) + \half \gamma^2\,\V (H^{i}(t))}=e^{ \sup (\gamma\,\log{x} + \gamma \, \E (H^{i}(t)) + \half \gamma^2\,\V (H^{i}(t)))} $. Thus  
  $ \pit^{*,\,i}$ $= \arg \sup (\gamma\,\log{x} + \gamma \, \E (H^{i}(t)) + \half \gamma^2\,\V (H^{i}(t)))= \arg \max           [ \gamma\,\log{x} + \gamma \E (\half \int_0^t( 2 \theta_s^i \pis^i \sigma_s - ( \pis^i \sigma_s)^2 )ds) +   \half \gamma^2\, \int_0^t (\pis^i \sigma_s)^2 ds] = \arg \max [\gamma (\half \E \int_0^t( 2 \theta_s^i \pis^i \sigma_s - ( \pis^i \sigma_s)^2 )ds) + \half \gamma^2\, \int_0^t (\pis^i \sigma_s)^2ds]=
  \arg \max \half\, \gamma\,[\E \int_0^t( 2 \theta_s^i \pis^i \sigma_s - ( \pis^i \sigma_s)^2 )ds +  \gamma\, \int_0^t (\pis^i \sigma_s)^2 ds] 
 \\ =  \arg \max [\E \int_0^t (2 \theta_s^i\,\pis^i \sigma_s - (1-\gamma)\, \int_0^t (\pis^i \sigma_s)^2 ds]= \arg \max [\E \int_0^t (2 \theta_s^i\,\pis^i \sigma_s - (1-\gamma)\,(\pis^i \sigma_s)^2)ds]\\ = \arg \max (1-\gamma)\E[\int_0^t (2 \frac{\theta_s^i}{1-\gamma}\,\pis^i \sigma_s -(\pis^i \sigma_s)^2)ds] = \arg \max \E[ \int_0^t(\frac{\theta_s^i}{1-\gamma})^2ds -\int_0^t ( \frac{\theta_s^i}{1-\gamma} - \pis^i \sigma_s )^2 ds ] .$
  Thus optimal is  achieved iff $\frac{\theta_s^i}{1-\gamma} - \pis^i \sigma_s =0$. Therefore
 \begin{equation*}
\pit^{*,\,i}=\frac{{\theta_\gamma}^i(t)}{\sigma_t}=\frac{\mut^i}{(1-\gamma)\sigma^2_t}=\frac{\mut+\upsilon^i_t \sigma_t}{(1-\gamma)\sigma^2_t}, \;\;\;t\in[0,T],
\end{equation*} 
where $\theta_\gamma^i= \frac{\theta^i}{1-\gamma}$ is the risk-adjusted Sharpe ratio of the stock for the $i$--th investor. 
It also follows that the value function for the $i$--th investor is 
\begin{equation*}
 u^i(x)=  e^{\gamma\,\log{x} + \half \gamma \,(1-\gamma) \E \int_0^t(\frac{\theta_s^i}{1-\gamma})^2ds }= e^{\gamma\,\log{x} + \half \,\frac{\gamma}{(1-\gamma)}\E \int_0^t(\theta_s^i)^2ds}. 
 \end{equation*} 
 \qedhere
\end{proof}
\begin{remark}
(1) Under power utility the optimal portfolio for each investor, relative to is filtration, is its risk-adjusted Sharpe ratio divided by the common volatility of the driving Brownian motion. It is random, being depended on the random processes $\upsilon^i_t$. Optimal portfolios are inversely proportional to the coefficient of RRA, $ 1-\gamma$ . When the RRA approaches 1, that is as $\gamma \rightarrow0$, we recover the optimal portfolio of the logarithmic utility, as seen in Guasoni's Theorem 3.1 above. However, as the RRA approach 0, the portfolios explode, while as the RRA becomes infinite, the portfolios become zero and all funds are invested in the risk--free asset. \\
(2) For both investors, the expected optimal portfolio is \textbf{deterministic} and inversely proportional to the RRA, and is independent of the mean-reversion speed and the mispricing level $q^2$. It explodes as the RRA approaches 0, that is, when $\gamma\rightarrow0$, and becomes essentially zero when either the RRA or stock price volatility becomes very large or infinite. 
\end{remark}

\begin{prop}\label{th2.22a}
Let $i\in\{0,1\}$,
and  $u^i(x)$ be the value function that solves (\ref{2.63}), where utility is the basis power function $ w^\gamma$ with RRA $1-\gamma$.  Let $\theta_\gamma^i= \frac{\theta^i}{1-\gamma}$ be the risk-adjusted Sharpe ratio of the stock for the $i$--th investor ,  and  risk--free rate $r=0$. \\
$(1)$ \hspace{.1in} The maximum expected power utility of terminal wealth for the $i$--th investor is 
\begin{eqnarray}
u^i(x) &= & e^{\gamma\,\log{x}+ \frac{\gamma}{2(1-\gamma)}\E \int_0^T (\theta^i_t)^2dt}= e^{\gamma\,\log{x}+ \frac{\gamma(1-\gamma)}{2}\E \int_0^T ((\theta^i_\gamma)_t)^2dt}\nonumber\\
              &=& e^{\gamma\,\log{x}+ \frac{\gamma}{2(1-\gamma)} \E \int_0^T \left(\frac{\mut+\upsilon^i_t \sigma_t}{\sigma_t}\right)^2dt}\label{2.65}.
\end{eqnarray}
$(2)$ \hspace{.1in} As $T\longrightarrow \infty$, the asymptotic maximum expected power utility is 
\begin{equation}\label{2.67}
u^i_{\infty}(x)\simeq e^{\gamma\,\log{x} + \frac{\gamma}{2(1-\gamma)}\left(\int_0^T \frac{\mut^2}{\sigma^2_t}\, dt + \frac{\lambda}{2}(1-p)(1+(-1)^{i+1}p)\,T\right)}.
\end{equation}
and  the (wealth) relative asymptotic maximum expected basis power utility is 
\begin{equation}\label{2.67}
\frac{u^1_{\infty}(x)}{ u^0_{\infty}(x)}= e^{\frac{\gamma}{2(1-\gamma)}\,\lambda\,p(1-p)\,T}.\\
\end{equation}
\end{prop}
\begin{proof} \textbf{ \{Proof of Proposition \ref{th2.22a}\}}\\
$(1)$
From Proposition \ref{pro2.22a}, the optimal portfolio is $\pit^{*,\,i}=\frac{\mut^i}{(1-\gamma)\sigma^2_t} =\frac{\mut+\upsilon^i_t \sigma_t}{(1-\gamma)\sigma^2_t}. $ Thus for the $i$--th investor, the value function or maximum expected power utility is $ u^i(x)= $
\begin{equation*}
 e^{\gamma\,\log{x} + \half \gamma \,(1-\gamma) \E \int_0^t(\frac{\theta_s^i}{1-\gamma})^2ds }= e^{\gamma\,\log{x} + \half \,\frac{\gamma}{(1-\gamma)}\E \int_0^t(\theta_s^i)^2ds} = e^{\gamma\,\log{x}+ \frac{\gamma}{2(1-\gamma)} \E \int_0^T \left(\frac{\mut+\upsilon^i_t \sigma_t}{\sigma_t}\right)^2dt}. 
 \end{equation*}
$(2)$. From Proposition \ref{pro2} ,  $ \int_0^T \E[\upsilon^i_t]^2 dt \rightarrow \frac{\lambda}{2}\,(1-p)(1+(-1)^{i+1}p)\,T,$ as $T\rightarrow \infty$. But $ \int_0^T \E(\frac{\mut+\upsilon^i_t \sigma_t}{\sigma_t})^2dt=  \int_0^T \E(\frac{\mut}{\sigma_t} + \upsilon^i_t)^2dt= \int_0^T(\frac{\mut}{\sigma_t})^2dt + \int_0^T \E (\upsilon^i_t)^2dt$ which converges to $ \int_0^T(\frac{\mut}{\sigma_t})^2dt + \frac{\lambda}{2}\,(1-p)(1+(-1)^{i+1}p)\,T\,\, $as $T\rightarrow \infty$.  
 Substituting this approximation into part $(1)$ yields, 
$ u^i_{\infty}(x)\simeq e^{\gamma\,\log{x} + \frac{\gamma}{2(1-\gamma)}\left(\int_0^T \frac{\mut^2}{\sigma^2_t}\, dt + \frac{\lambda}{2}(1-p)(1+(-1)^{i+1}p)\,T\right)}.$
Imposing Part $(3)$ of Proposition \ref{pro2} and simplifying the the quotient of the value function of the informed and the uninformed investors yield the desired result.
\end{proof}
\begin{remark}
(1) The maximum expected power utility from terminal wealth grows exponentially with the natural logarithm initial wealth $x$ adjusted by one minus the RRA.\\
(2) Observe that the wealth relative asymptotic power utility is an exponential function of the \textbf{risk-adjusted} excess asymptotic utility under logarithmic utility, which is independent of the initial wealth of the investors. This is an important unexpected result! 
 \end{remark}
 We now present our main result for the general power utility function, which is analogous to the previous theorem.
 \begin{thm}[\textbf{Main}]\label{th2.22b}
Let $i\in\{0,1\}$,  $u^i(x)$ be the value function, and $\pi^{*,\,i}$ the optimal portfolio for the $i$-th investor that solves (\ref{2.63}), where utility is assumed to be general power utility function $ \U_\gamma(w)=\frac{w^\gamma -1}{\gamma}, \,\, \gamma <1   $ and the risk--free interest rate is $r=0$. Define the log-linear value function $\psi^{i}(x)= \log(1+\gamma\,u^{i}(x))$, where $x$ is the initial wealth and $1-\gamma$ is the relative risk aversion.  Let $\theta_\gamma^i= \frac{\theta^i}{1-\gamma}$ be the risk-adjusted Sharpe ratio of the stock for the $i$--th investor ,  and  risk--free rate $r=0$. \\
$(1)$ \hspace{.1in}  The optimal portfolios are given by Proposition \ref{pro2.22a}.\\
$(2)$ \hspace{.1in} The log-linear maximum expected utility of terminal wealth is 
\begin{eqnarray}
\psi^i(x)&=& \gamma\,\log{x}+ \frac{\gamma}{2(1-\gamma)}\E \int_0^T (\theta^i_t)^2dt = \gamma\,\log{x}+ \frac{\gamma(1-\gamma)}{2}\E \int_0^T (\theta_\gamma^i)_t^2dt\nonumber\\
   &=& \gamma\,\log{x}+ \frac{\gamma}{2(1-\gamma)} \E \int_0^T \left(\frac{\mut+\upsilon^i_t \sigma_t}{\sigma_t}\right)^2dt.\label{2.65}
\end{eqnarray}
$(3)$\hspace{.1in} As $T\longrightarrow \infty$, the asymptotic log-linear maximum expected  power utility is 
\begin{equation}\label{2.67}
\psi^i_{\infty}(x)=\gamma\,\log{x} + \frac{\gamma}{2(1-\gamma)}\left( \int_0^T \frac{\mut^2}{\sigma^2_t}\, dt + \frac{\lambda}{2}(1-p)
(1+(-1)^{i+1}p)\,T\right)\\
\end{equation}
$(4)$\hspace{.1in}
The excess asymptotic  log-linear maximum expected utility is 
\begin{equation}\label{2.68}
\psi^1_{\infty}(x)- \psi^0_{\infty}(x) \simeq \frac{\gamma}{2(1-\gamma)}\lambda\,p(1-p)\,T, 
\end{equation}
which vanishes as the RRA  tends to 1. \\
$(5)$ Moreover, the excess asymptotic maximum expected utility is 
\begin{equation}\label{2.68a}
u^1_{\infty}(x)- u^0_{\infty}(x) \simeq \frac{1}{2}\lambda\,p(1-p)\,T, 
\end{equation}
\end{thm}
\begin{proof} \textbf{ \{Proof of Main Theorem \}} \\
$(1)$\, Since the utility function is a linear transformation of that the utility function used in the proof of  Proposition \ref{pro2.22a}, the optimal portfolio is identical for each respective investor, and the proof is the same.\\
$(2)$\, From Proposition \ref{pro2.22a}, the value function for the general power utility function is 
$u^i(x) =  \frac{e^{\gamma\,\log{x}+ \frac{\gamma}{2(1-\gamma)}\E \int_0^T (\theta^i_t)^2dt} - 1 }{ \gamma} = \frac{e^{\gamma\,\log{x}+ \frac{\gamma(1-\gamma)}{2)}\E \int_0^T (\theta_\gamma^i)_t^2dt} - 1 }{ \gamma}.$\\ 
Rearranging, the log-linear value function for each investor is \\
$ \psi^i(x)= \log(1+\gamma\,u^i(x))= \gamma\,\log{x}+ \frac{\gamma}{2(1-\gamma)}\E \int_0^T (\theta^i_t)^2 dt = \gamma\,\log{x}+ \frac{\gamma}{2(1-\gamma)} \E \int_0^T \left(\frac{\mut+\upsilon^i_t \sigma_t}{\sigma_t}\right)^2dt.$\\
$(3)$ Fix $\gamma$, and let $T\rightarrow\infty$. Then $ \E \int_0^T \left(\frac{\mut+\upsilon^i_t \sigma_t}{\sigma_t}\right)^2dt \rightarrow  \int_0^T \frac{\mut^2}{\sigma^2_t}\, dt + \frac{\lambda}{2}(1-p)
(1+(-1)^{i+1}p)\,T,$
whence $\psi^i(x)\rightarrow \gamma\,\log{x} + \frac{\gamma}{2(1-\gamma)}\left( \int_0^T \frac{\mut^2}{\sigma^2_t}\, dt + \frac{\lambda}{2}(1-p)
(1+(-1)^{i+1}p)\,T\right). $ \\
$(4)$\, The excess asymptotic log-linear maximum expected power utility follows  from$(3)$.\\
$(5)$\, Let $\gamma\rightarrow0$. From $(2)$ $\frac{\psi^i(x)}{\gamma}=\lim_{\gamma\rightarrow0} \log{x} + \frac{1}{2(1-\gamma)}\left( \int_0^T \frac{\mut^2}{\sigma^2_t}\, dt + \frac{\lambda}{2}(1-p)
(1+(-1)^{i+1}p)\,T\right)\\= \log{x} + \frac{1}{2}\left( \int_0^T \frac{\mut^2}{\sigma^2_t}\, dt + \frac{\lambda}{2}(1-p)
(1+(-1)^{i+1}p)\,T\right).$
 But we also have
 $\lim_{\gamma\rightarrow0}\frac{\psi^i(x)}{\gamma}=\lim_{\gamma\rightarrow0}\frac{\log(1+\gamma\,u^i(x))}{\gamma}=u^i(x)$, which yields the required result by differencing.
\end{proof}

\section{Conclusion}
We extend the theory of asymmetric information in mispricing models of
stocks with prices that follow geometric Brownian motion to investors
having CRRA preferences. We obtain optimal portfolios and utilities,
including asymptotic utilities for informed and uninformed investors
under the general power log--linear power utility. Our results contain
those of Guasoni (2006) as a special case of the relative risk
aversion approaching one.


\section*{References}

Amendinger, J., P.~Imkeller, and M.~Schweizer (1998).
\newblock Additional logarithmic utility of an insider.
\newblock {\em Stochastic Processes and their Application\/}~{\em {\bf 75}},
263--286.

Applebaum, D. (2004).
\newblock {\em L\'{e}vy {P}rocesses and {S}tochastic {C}alculus}.
\newblock Cambridge University Press.

Barberis, N., A.~Shleifer, and R.~Vishny (1998).
\newblock A model of investor sentiment.
\newblock {\em Journal of Financial Economics\/}~{\em {\bf 49}}, 307--43.

Brunnermeier, M.~K. (2001).
\newblock {\em Asset Pricing Under Asymmetric Information: Bubbles, Crashes,
  Technical Analysis and Herding}.
\newblock Oxford University Press.

Buckley, W. (2009).
\newblock {\em Asymmetric information in fads models in L\'{e}vy markets}.
\newblock Ph.\ D. thesis, Florida Atlantic University.

Campbell, J.~Y. and A.~S. Kyle (1993).
\newblock Smart {M}oney, {N}oise {T}rading and {S}tock {P}rice {B}ehaviour.
\newblock {\em Review of Economic Studies\/}~{\em {\bf 60}\/}(1), 1--34.

Campbell, J.~Y. and R.~J. Shiller (1987).
\newblock Cointegration and {T}ests of {P}resent {V}alue {M}odels.
\newblock {\em Journal of Political Economy\/}~{\em {\bf 95}\/}(5), 1062--1088.

Cheridito, P. (2003).
\newblock Representation of {G}aussian measures that are equivalent to {W}iener
  measure.
\newblock {\em Seminar de Probabilities\/}~{\em {\bf XXXVII}}, 81--89.
\newblock Springer Lecture Notes in Mathematics 1832.

Cutler, D., J.~Poterba, and L.~Summers (1990).
\newblock Speculative dynamics and the role of feedback traders.
\newblock {\em American Economic Review\/}~{\em {\bf 80}}, 63--68.

Daniel, K., D.~Hirshleifer, and A.~Subrahmanyam (1998).
\newblock Investor {P}sychology and {S}ecurity {M}arket {U}nder- and
  {O}verreactions.
\newblock {\em Journal of Finance\/}~{\em {\bf 53}}, 1839--85.

{De Long}, J.~B., A.~Shleifer, L.~H. Summers, and R.~J. Waldmann (1990).
\newblock Noise {T}rader {R}isk in {F}inancial {M}arkets.
\newblock {\em Journal of Political Economy\/}~{\em {\bf 98}}, 703--73.

Fama, E. (1970).
\newblock Efficient {C}apital {M}arkets: {A} {R}eview of {T}heory and
  {E}mpirical {W}ork.
\newblock {\em Journal of Finance\/}~{\em {\bf 25}}, 383--417.

Flavin, M.~A. (1983).
\newblock Excess {V}olatility in the {F}inancial {M}arkets: A {R}eassessment of
  the {E}mpirical {E}vidence.
\newblock {\em Journal of Political Economy\/}~{\em {\bf 91}\/}(6), 929--956.

Guasoni, P. (2006).
\newblock Asymmetric information in fads models.
\newblock {\em Finance and Stochastics\/}~{\em {\bf 10}}, 159--177.

Hitsuda, M. (1968).
\newblock Representation of {G}aussian processes equivalent to {W}iener
  process.
\newblock {\em Osaka Journal of Mathematics\/}~{\em {\bf 5}}, 299--312.

Imkeller, P. and S.~Ankirchner (2006, August).
\newblock Financial markets with asymmetric information, information drift,
  additional utility and entropy.
\newblock Technical report, Humboldt University.

Karatzas, I. and I.~Pikovsky (1996).
\newblock Anticipative portfolio optimization.
\newblock {\em Advances in Applied Probability\/}~{\em {\bf 28}}, 1095--1122.

Karatzas, I. and S.~Shreve (1991).
\newblock Martingale and duality methods for utility maximization in an
  incomplete market.
\newblock {\em SIAM Journal on Control and Optimization\/}~{\em {\bf 29}\/}(3),
  702--730.

Karatzas, I. and S.~Shreve (1999).
\newblock {\em Brownian Motion and Stochastic Calculus\/} (2nd ed.).
\newblock Springer.

Kleidon, A.~W. (1986).
\newblock Variance {B}ounds {T}ests and {S}tock {P}rice {V}aluation {M}odels.
\newblock {\em Journal of Political Economy\/}~{\em {\bf 94}}, 953--1001.

LeRoy, S.~F. and R.~D. Porter (1981).
\newblock The {P}resent-{V}alue {R}elation: {T}ests {B}ased on {I}mplied
  {V}ariance {B}ounds.
\newblock {\em Econometrica\/}~{\em {\bf 49}\/}(3), 555--574.

March, T. and R.~Merton (1986).
\newblock Dividend variability and variance bounds tests for the rationality of
  stock market prices.
\newblock {\em American Economic Review\/}~{\em {\bf 76}}, 483--98.

Merton, R.~C. (1971).
\newblock Optimum consumption and portfolio rules in a continuous-time model.
\newblock {\em Journal of Economic Theory\/}~{\em {\bf 3}}, 373--413.

Odean, T. (1998).
\newblock Volume, {V}olatility, {P}rice, and {P}rofit {W}hen {A}ll {T}raders
  {A}re {A}bove {A}verage.
\newblock {\em Journal of Finance\/}~{\em {\bf 53}}, 1887--1934.

Oksendal, B. (2005).
\newblock {\em Stochastic Integral Equations\/} (6th ed.).
\newblock Springer.

Protter, P. (2004).
\newblock {\em Stochastic Integration and Differential Equations\/} (2nd ed.).
\newblock Springer.

Shiller, R. (1981).
\newblock Do stock prices move to much to be justified by subsequent changes in
  dividends?
\newblock {\em American Economic Review\/}~{\em {\bf 71}}, 421--436.

Summers, L.~H. (1986).
\newblock Does the {S}tock {M}arket {R}ationally {R}eflect {F}undamental
  {V}alues?
\newblock {\em Journal of Finance\/}~{\em {\bf 41}}, 591--601.

Summers, L.~H. and J.~M. Poterba (1988).
\newblock Mean reversion in stock prices: {E}vidence and implications.
\newblock {\em Journal of Financial Economics\/}~{\em {\bf 21}\/}(1), 27--59.

Wang, J. (1993).
\newblock A {M}odel of {I}ntertemporal {A}sset {P}rices {U}nder {A}symmetric
  {I}nformation.
\newblock {\em Review of Economic Studies\/}~{\em {\bf 60}}, 249--282.

West, K.~D. (1988).
\newblock Dividend {I}nnovations and {S}tock {P}rice {V}olatility.
\newblock {\em Econometrica\/}~{\em {\bf 56}}, 37--61.

\appendix
\section{}
\begin{proof} \textbf{ \{Proof of Proposition 1:\}}
From equations (\ref{2.2}) and (\ref{2.3}), we have for the informed investor ($i=1$) 
$dY_t = p\,dW_t + q\,dU_t
= p\,dW_t + q\,d(-\lambda U_t dt + dB_t)= p\,dW_t + q\,dB_t -q\,\lambda U_t dt 
= dB^1_t -q\,\lambda U_t dt 
= dB^1_t +\upsilon_t^1 dt,$
where
\begin{equation}
B^1_t  \stackrel{\triangle}{=}  p\,W_t + q\,B_t,\,\,\upsilon^1_t \stackrel{\triangle}{=}-q\,\lambda U_t.\label{2.18a}
\end{equation}
Substituting (\ref{2.18a}) into (\ref{2.1}), yields
$\frac{dS_t}{S_t} = \mut dt+ \sigma_t dY_t
= \mut dt+ \sigma_t (dB^1_t +\upsilon_t^1 dt)\\= (\mut+ \upsilon^1_t \sigma_t)dt + \sigma_t dB^1_t\nonumber 
=\mut^1\,dt + \sigma_t dB^1_t,$
where
\begin{equation}\label{2.19}
\mut^1\stackrel{\triangle}{=}\mut+ \upsilon^1_t \sigma_t.
\end{equation}
So under $\calF^1$, the informed investor has stock price 
\begin{equation}\label{2.21}
S_t= S_0\,\exp{\left( \int_0^t(\mu^1_s - \half \sigma^2_s)ds+ \int_0^t\sigma_s dB^1_s\right)},\:\:\:t\in [0,\,T],
\end{equation}
where $B^1$ is an $\calF^1$--Brownian motion given by  (\ref{2.18a}) and  $\mut^1$ is given by (\ref{2.19}).\\ 
For the uninformed investor($i=0$), we have  $dY_t=dB^0_t + \upsilon_t^0dt  $ from Equation \ref{2.30}, where $\mut^0=\mut + \upsilon^0_t \sigma_t. $ Substituting into (\ref{2.1}) yields the desired result. 
\end{proof}
\noindent \textbf{Proof of Theorem 1}:\,
This follows below 
with $\mut^{i}, \theta_t^{i}, \pit^{i} , V_t^{i} $, etc, replacing the relevant variables,  and $B_t^i$ replacing $Y_t$.    

\begin{proof}
It is clear that 
$\frac{dV_t}{V_t}=(1-\pit)\, r_t \,dt + \pit \,\frac{dS_t}{S_t}.$
Imposing the stock price dynamic (\ref{2.1a}) on the last equation, yields
$\frac{dV_t}{V_t}$
$=(1-\pit) r_t \,dt + \pit \,(\mut dt+ \sigma_t \,dY_t )
=(r_t +\pit( \mut - r_t))\,dt + \pit \,\sigma_t\, dY_t .$
Thus, we get the percentage returns dynamic of the wealth process
\begin{equation}\label{2.40}
\frac{dV_t}{V_t}=(r_t +\pit \,\sigma_t\,\theta_t)\,dt + \pit \,\sigma_t\, dY_t ,
\end{equation}
where $\theta$, is the stock's Sharpe ratio or market price of risk. 
By the It\^{o} formula, the unique solution to (\ref{2.40}) with $V_0=x$, is the stochastic exponential
\begin{equation}\label{2.41}
V_t= x\, \exp \left(  \int_0^t r_s\, ds + \int_0^t (\pis \,\sigma_s\, \theta_s -\half \,\pis^2\, \sigma^2_s)\,ds + \int_0^t \pis \sigma_s \,dY_s  \right),
\end{equation}
with discounted wealth process $\widetilde{V}$ given by 
\begin{equation*}\label{2.42}
\widetilde{V}_t\stackrel{\triangle}{=}\exp{(-\int_0^t\, r_s \,ds)}\,V_t=x\, \exp \left( \int_0^t (\pis \,\sigma_s\, \theta_s -\half\, \pis^2\, \sigma^2_s\,)ds + \int_0^t \pis\, \sigma_s \,dY_s \right).
\end{equation*}
The power utility of discounted wealth is
$\U({\widetilde{V}_t})= e^{\gamma\,\log{\widetilde{V}_t}}= = e^{\gamma\,\log{x} + \gamma\,H(t)},$
where
$\log{\widetilde{V}_t}=\log{x} +  \int_0^t (\pis \sigma_s \theta_s -\half \pis^2 \sigma^2_s)ds + \int_0^t \pis \sigma_s dY_s 
=\log{x} + \half \int_0^t \theta^2_sds  -\half \int_0^t (\pis \sigma_s - \theta_s)^2 ds + \int_0^t \pis \sigma_s dY_s
=\log(x) + H(t).$ \qedhere
\end{proof}
\end{document}